\documentclass[11pt,fleqn,a4paper]{article}
\pdfoutput=1
\usepackage{amsmath,amssymb,amsfonts,amsthm}
\usepackage{siunitx,hyperref,lipsum}
\usepackage{xfrac}
\usepackage{url}
\usepackage{textcomp}
\usepackage{xspace}
\usepackage{enumerate}

\usepackage{algorithm}
\usepackage[noend]{algpseudocode}
\usepackage{comment}
\usepackage{xcolor}

\newtheorem{theorem}{Theorem}
\newtheorem{reduce}{Rule}
\newtheorem{lemma}{Lemma}
\newtheorem{proposition}{Proposition}

\newtheorem{definition}{Definition}

\newcommand{\cO}{{O}}
\newcommand{\cOs}{{O^{*}}}
\renewcommand{\emptyset}{\varnothing} 
\DeclareMathOperator*{\argmin}{arg\,min}

\newcommand{\defproblemu}[3]{
  \vspace{0.5mm}
  \noindent\fbox{
  \begin{minipage}{.95\textwidth}
        #1 \\
      {\bf{Instance:}} #2  \\
      {\bf{Task:}} #3
    \end{minipage}
  }
  \vspace{0.5mm}
}

\newcommand{\hcd}{\textsc{Highly Connected Deletion \xspace}}
\newcommand{\phcd}{\(p\)-\textsc{Highly Connected Deletion \xspace}}
\newcommand{\hcs}{\textsc{Isolated Highly Connected Subgraph \xspace}}
\newcommand{\fhcs}{$f$-\textsc{Isolated Highly Connected Subgraph \xspace}}

\newcommand{\seed}{\textsc{Seeded Highly Connected Edge Deletion \xspace}}

\begin{document}
\title{Parameterized Algorithms for Partitioning Graphs into Highly Connected Clusters}
\author{Ivan Bliznets\thanks{St. Petersburg Department of V.A. Steklov Institute of Mathematics of the Russian Academy of Sciences, Russia,\texttt{iabliznets@gmail.com}.}
	 \and Nikolai Karpov\thanks{St. Petersburg Department of V.A. Steklov Institute of Mathematics of the Russian Academy of Sciences, Russia,\texttt{kimaska@gmail.com}.}}
\maketitle

\begin{abstract}
Clustering is a well-known and important problem with numerous applications. The graph-based model is one of the typical cluster models. In the graph model, clusters are generally defined as cliques. However, such an approach might be too restrictive as in some applications, not all objects from the same cluster must be connected. That is why different types of cliques relaxations often considered as clusters.

  In our work, we consider a problem of partitioning graph into clusters and a problem of isolating cluster of a special type where by cluster we mean highly connected subgraph. Initially, such clusterization was proposed by Hartuv and Shamir. And their  HCS clustering algorithm was extensively applied in practice. It was used to cluster cDNA fingerprints, to find complexes in protein-protein interaction data, to group protein sequences hierarchically into superfamily and family clusters, to find families of regulatory RNA structures. The HCS algorithm partitions graph in highly connected subgraphs. However, it is achieved by deletion of not necessarily the minimum  number of edges. In our work, we try to minimize the number of edge deletions. We consider problems from the parameterized point of view where the main parameter is a number of allowed edge deletions. The presented algorithms significantly improve previous known running times for the \hcd (improved from  $\cOs\left(81^k\right)$ to  $\cOs\left(3^k\right)$), \hcs(from $\cOs(4^k)$ to $\cOs\left(k^{\cO\left(k^{\sfrac{2}{3}}\right)}\right)$ ), \seed (from $\cOs\left(16^{k^{\sfrac{3}{4}}}\right)$ to $\cOs\left(k^{\sqrt{k}}\right)$) problems. Furthermore, we present a subexponential algorithm for \hcd problem if the number of clusters is bounded. Overall our work contains three subexponential algorithms which is unusual as very recently there were known very few problems admitting subexponential algorithms.

\end{abstract}

\section{Introduction}

Clustering is a problem of grouping objects such that objects in one group are more similar to each other than to objects in other groups. Clustering has numerous applications, including: machine learning,  pattern recognition, image analysis, information retrieval, bioinformatics, data compression, and computer graphics. Graph-based model is one of the typical cluster models.  In a graph-based model most commonly cluster is defined as a clique. However,  in many applications, such definition of a cluster is too restrictive~\cite{{pattillo2013clique}}. Moreover, clique model generally leads to computationally hard problems. For example clique problem is $W[1]-hard$ while $s$-club problem,
with $s\geq2$, is fixed-parameter tractable with respect to the parameters solution size and $s$~\cite{schafer2009exact}. Because of the two mentioned reasons researchers consider different clique relaxation models~\cite{pattillo2013clique, shahinpour2013distance}.  We mention just some of the  possible relaxations: $s$-club(the diameter is less than of equal to $s$), $s$-plex (the smallest degree is at least $|G|-s$), $s$-defective clique (missing $s$ edges to complete graph), $\gamma$-quasi-clique ($|E|/{ {|V|}\choose{2}} \geq \gamma$), highly connected graphs (smallest degree bigger than $|G|/2$) and others. With different degree of details all these relaxations were studied:  $s$-club\cite{schafer2009exact, shahinpour2013distance}, $s$-plex \cite{moser2009algorithms, balasundaram2011clique}, $s$-defective clique~\cite{yu2006predicting, guo2011editing}, $\gamma$-quasi-clique~\cite{PattilloYB13, pattillo2013maximum}, highly connected graphs \cite{HuffnerKS15, HuffnerKLN14, HartuvS00}. 

In this work, we study the clustering problem based on highly connected components model. A graph is \emph{highly connected} if the edge connectivity of a graph(the minimum number of edges whose deletion results in a disconnected graph) is bigger than $\frac{n}{2}$ where $n$ is the number of vertices in a graph. An equivalent characterization is for each vertex has degree bigger than $\frac{n}{2}$, it was proved in \cite{graphn2}. One of the reasons for this choice is a huge success in applications of the  Highly Connected Subgraphs(HCS) clustering algorithm proposed by Hartuv and Shamir and the second reason is the lack of research for this model compared with the standard clique model. HCS algorithm was used~\cite{HuffnerKLN14} to cluster cDNA fingerprints~\cite{hartuv2000algorithm}, to find complexes in protein-protein interaction data~\cite{hayes2013graphlet}, to group protein sequences hierarchically into superfamily and family clusters~\cite{krause2005large}, to find families of regulatory RNA structures~\cite{parker2011new}.

H{\"{u}}ffner et al.~\cite{HuffnerKLN14} noted that while Hartuv and Shamir’s algorithm partitions a graph into highly connected components, it does not delete the minimum number of edges required for such partitioning. That is why they initiated study of the following problem 


\defproblemu{\hcd}{Graph $G = (V, E)$.}{
	Find edge subset  $E' \subseteq E$ of the minimum size  such that 
	each connected component of $G' = (V, E \setminus E')$ is highly connected.}
	

For this problem, H\"{u}ffner et al.~\cite{HuffnerKLN14} proposed an algorithm which is based on the dynamic programming technique 
with the running time bounded by \(\cOs(3^n)\) where $n$ is the number of vertices. For parameterized version of the problem they proposed an algorithm with the running time \(\cOs(81^k)\) where \(k\) is an upper bound on the size of \(E'\). 
Additionally, they proved that the problem admits a kernel with the size \(\cO(k^{1.5})\). Moreover, 
they proved conditional lower bound on the running time of algorithms for \hcd, in particular, the problem cannot be solved in time
\(2^{o(k)}{\cdot}n^{O(1)}, 2^{o(n)}{\cdot}n^{O(1)}\,,\) or \(2^{o(m)}{\cdot}n^{O(1)}\) unless 
the exponential-time hypothesis (ETH) fails. 

Moreover, in another work H\"{u}ffner et al.~\cite{HuffnerKS15} studied a parameterized complexity of related problem of finding highly connected components in a graph.

\defproblemu{\hcs}{Graph $G = (V, E)$, integer $k$, integer $s$.}{
	Is there a set of vertices $S$ such that $|S|=s$, $G[S]$ is highly connected graph 
	and $|E(S, V \setminus S)| \leq k$.
}

\defproblemu{\seed}{
	Graph $G = (V, E)$, subset  $S \subseteq V$, integer $a$, integer $k$.
}{
Is there a subset of edges $E' \subseteq E$ of size at most $k$  such that $G - E'$ contains only isolated vertices and one highly connected component $C$ with $S \subseteq V(C)$ and $|V(C)| = |S| + a$.
}

They proposed algorithms with the running time \( \cOs(4^k) \) and \( \cOs(16^{k^{3/4}}) \) respectively.

\textbf{Our results:}
We propose algorithms which significantly improve previous upper bounds. Running times of algorithms may be found in a Table~\ref{tab:res}. We would like to note that three of the algorithms have subexponential running time which is not common. Until very recently there were very few  problems admitting subexponential running time. To our mind in algorithm for \hcs problem we have an unusual branching procedure as in one branch parameter is not decreasing. However, the value of subsequent decrementation of parameter in this branch is increasing which leads to subexponential running time.  We find the fact interesting as we have not met such behavior of branching procedures before. Presented analysis for this case might be useful in further development of subexponential algorithms.   

\begin{table}
\begin{center}
\begin{tabular}{|c|c|c|}
\hline
Problem & Previous result & Our result \\ \hline
\hcd (exact) & $\cOs\left(3^n\right)$ & $\cOs\left(2^n\right)$ \\ \hline
\hcd (parameterized) & $\cOs\left(81^k\right)$ & $\cOs\left(3^k\right)$ \\ \hline
\phcd & - & $\cOs\left(2^{\cO\left(\sqrt{pk}\right)}\right)$ \\  \hline
\hcs & $\cOs(4^k)$ & $\cOs\left(k^{\cO\left(k^{\sfrac{2}{3}}\right)}\right)$ \\ \hline
\seed & $\cOs\left(16^{k^{\sfrac{3}{4}}}\right)$ & $\cOs\left(k^{\sqrt{k}}\right)$ \\ \hline

\end{tabular}
\end{center}
\caption{Results}\label{tab:res}
\end{table}

\section{Algorithms for partitioning} 
\subsection{\hcd}
In this section we present an algorithm for \hcd problem. Our algorithm is based on the fast 
subset convolution. Let $f, g: 2^X \to \{0,1,\dots M\}$ be two functions and $|X|=n$.
Bj{\"o}rklund et al. in \cite{BjorklundHKK07} proved that function  $f \ast g: 2^X \to \{0,\dotsc, 2M\}$, where 
 $(f \ast g)(S)  = \min\limits_{T \subseteq S}\left(f(T) + g(S \setminus T)\right)$, can be computed on all subsets $S\subseteq X$ in time  $\cO(2^n\mathrm{poly}(n,M))$.

\begin{theorem}\label{exact:hcd}
There is a $\cOs(2^n)$ time  algorithm for \hcd problem.
\end{theorem}

\begin{proof}
	Let  define function $f$ in the following way 
	\begin{equation*}
		f(S)=
		\left\{
		\begin{array}{ll}
			|E(S, V \setminus S)| & \mbox{if }G[S]  \mbox{ is highly connected} \\
			\infty & \mbox{otherwise}
		\end{array}
		\right.
	\end{equation*}
	Consider function $f^{\ast k}(V)= \underbrace{f*\dots*f}_\text{k times}$.
	Note that $f^{\ast k}(V) = \min\limits_{S_1 \sqcup \dotsb \sqcup S_k = V}\left(f(S_1) + \dotsb + f(S_k)\right)$. Hence, to solve the problem it is enough to find minimum of  $f^{\ast k}(V)$ over all $1\leq k \leq n$.
	Note that if $f^{\ast k}(V)=\infty$ then it is not possible to partition $V$ into $k$ highly connected components. So if the minimum value of $f^{\ast k}(V)$  is $\infty$ then there is no partitioning of $G$ into highly connected components. 
	
	Our algorithm contains the following steps.
	\begin{enumerate}
		\item Compute $f$, i.e. compute value $f(S)$ for all $S \subset V$. It takes $\cO(2^n(n+m))$ time.
		\item Using Bj{\"o}rklund et al.\cite{BjorklundHKK07} algorithm iteratively compute $f^{\ast i}$ for all $1 \leq i \leq n$. 
		\item  Find $k$ such that $f^{\ast k}(V)$ is minimal. 
	\end{enumerate}

	After we perform above steps  we will know values of functions $f^{\ast i}$ on each subset $S\subseteq X$.
	Let $S_1 \sqcup S_2 \sqcup \dots \sqcup S_k$ be an optimum partitioning of $X$ into highly connected components.
	Knowing values of function   $f^{\ast {k-1}}$  and $f$ it is straightforward to restore $S_k$ in time $2^n$.  
	Moreover, knowing $f^{\ast {k-1}}, S_k$ we can find value of $S_{k-1}$. Proceeding this way we obtain the optimum partitioning. As  $ k \leq n$, we spent at most $\cO(n2^{n})$ time to find all $S_i$.
	
	It is left to show how to compute all $f^{\ast{i}}$ within \(\cOs(2^{n})\) time. The only obstacle why we cannot straightforwardly  apply Bj{\"o}rklund's algorithm is that $f$ sometimes takes infinite value. It is easy to fix the problem by replacing infinity value with $2m + 1$. We know that  each convolution require  
	\( \cO(2^{n}{\mathrm{poly}}(n, M)) \) time and above we show that we can put $M$ to be equal $2m+1$. As we need to perform 
	$n$  subset convolutions. So, the running time of second step is $\cOs(2^n)$. Hence, the overall running time is $\cOs(2^n)$.
\end{proof}

Now we consider parameterized version of \hcd problem (one is asked whether it is possible to delete at most $k$ edges  and get a vertex disjoint union of highly connected subgraphs).

\begin{theorem}\label{parameterized:hcd}
There is an algorithm for \hcd problem with running time  $\cOs(3^k)$.
\end{theorem}
\begin{proof}

Before we proceed with the proof of the theorem we list several simplification rules and 
lemmas proved by  H{\"{u}}ffner et al. in \cite{HuffnerKLN14}.

\begin{reduce}
If $G$ contains a connected component  $C$ which is highly connected then replace original instance with 
instance $(G[V \setminus V(C)], k)$.
\end{reduce}

\begin{lemma}\label{neigh}
Let $G$ be a highly connected graph and  $u, v \in V(G) $ be two different vertices from $V(G)$. 
If $uv \in E$, then $|N(u) \cap N(v)| \geq 1 $. If $uv \not\in E$ then $|N(u) \cap N(v)| \geq 3$.
\end{lemma}

\begin{reduce}
If  $u, v \in E$ and $N(u) \cap N(v) =\emptyset$ then delete edge $uv$ and decrease parameter $k$ by 1. 
The obtained instance is $((V, E \setminus \{uv\}), k - 1)$.
\end{reduce}

\begin{definition}
Let us call vertices $u, v$ \emph{$k$-connected} if any cut separating these two vertices has size bigger 
than $k$.
\end{definition}
\begin{reduce}\label{inseparable}
Let  $S$ be an inclusion maximal set of pairwise $k$-connected vertices and $|S|>2k$. If the induced graph 
$G[S]$ is not highly connected then our instance is a  NO-instance(it is not possible to delete $k$ edges and obtain  vertex disjoint union of higly connected subgraphs). Otherwise, we replace original instance with an instance  $(G[V \setminus S], k - |E(S, V \setminus S)|)$.
\end{reduce}
\begin{lemma}\label{lemma-dist}
If $G$ is highly connected then $diam(G)\leq 2$.
\end{lemma}

It was shown in \cite{HuffnerKLN14}  that all of the above rules are applicable in polynomial time.

Without loss of generality assume that $G$ is connected. Otherwise, we  consider 
several independent problems. One problem for each connected component.   For each connected component
we find minimum number of edges that we have to delete in order to partition this component into 
highly connected subgraphs. Note that in order to find a minimum number for each subproblem 
we simply consider all possible values of parameter starting from $0$ to $k$.

From Lemma~\ref{lemma-dist} follows that if  $dist(u,v)$ (distance between two vertices $u,v$) is bigger than $2$ then in optimal partitioning $u$ and $v$ belong to different connected components. 
Hence, if $dist(u,v)\geq 3$ then at least one edge from the shortest path between $u$ and $v$ belongs to $E'$.
If $diam(G)>2$ then it is possible to find two vertices $u,v$ such that $dist(u,v)=3$. 
So given the shortest path $u, x, y, v$ we can branch to three instances  $(G\setminus ux, k - 1)$,
$(G\setminus xy, k - 1)$, $(G\setminus yv, k - 1)$. We apply such branching exhaustively. Finally, we obtain instance with a graph $G'$ of diameter 2. 

Now, for our algorithm it is enough to consider a case when graph $G$ has the following properties:
(i) $diam(G) \leq 2$; (ii) there are no subsets $S$ of pairwise $k$-connected vertices with $|S|>2k$; (iii) $G$ is not highly connected. 


From now on we assume that $G$ has above mentioned properties.
Suppose $C_1\sqcup C_2 \sqcup \dots \sqcup C_\ell$ is an optimum partitioning of $G$ into 
highly connected graphs and  $E'$ is a subset of removed edges. 
We call vertex \emph{affected} if it is incident with an edge from $E'$. 
Otherwise, it is \emph{unaffected}.
Denote by $U$ the set of all unaffected vertices and by $T$ the set  of all affected vertices. By $C(v)$ we denote a cluster $C_i$ for which $v \in C_i$.  Note that for affected vertex $u$ there is vertex $v$ such that $uv \in E(G)$ and $v \notin C(u)$.

\begin{lemma}\label{alone}
Let $G$ be a graph with diameter $2$ then for any optimum partitioning 
$C_1\sqcup C_2 \sqcup \dots \sqcup C_\ell$ of $G$ into highly connected graphs  there is an $i$ such that 
$U$ is contained in $C_i$.  
\end{lemma}
\begin{proof}
Assume that there are two unaffected vertices $u, v  \in U$ and $C(v)\neq C(u)$. 
Note that any path between $u$ and $v$ must contain an edge from $E'$ and two different edges contained in 
$C(u), C(v)$ and incident to $u$ and $v$ correspondingly. 
So, the shortest path between $u$ and $v$ contains at least three edges which 
contradict our assumption that $diam(G)\leq 2$. Hence, there is an $i$ such that $U \subseteq C_i$. 
\end{proof}

\begin{lemma}\label{cluster}
Let $G$ be a graph with diameter $2$ and optimum partitioning $C_1 \sqcup C_2\sqcup \dots \sqcup C_\ell$ into 
highly connected graphs. If $U$ is not empty then $|E'| \geq n - |C_i|$ 
where  $U \subseteq C_i$.
\end{lemma}
\begin{proof}
Consider an arbitrary unaffected vertex $u$. For any $v \in V$ we have $dist(v,u)\leq 2$. 
Hence, for any $v\notin C(u)$ there is an edge connecting component $C(u)$ with vertex $v$ as 
otherwise we have $dist(u,v)>2$. So we have  $|E'|\geq n - |C(u)|$.
\end{proof}

For any \texttt{YES}-instance we have  $k \geq |E'| \geq \frac{|T|}{2}$, $n = |T| + |U|$, and $|U| \leq 2k$.The inequality $|U| \leq 2k$ follows from the simplification Rule~\ref{inseparable} and Lemma~\ref{alone}. As otherwise highly connected component which contains $U$ is bigger than $2k$ and 
hence simplification Rule~\ref{inseparable} can be applied which leads to contradiction. So, it means that $n = |T| + |U| \leq 4k$.

Below we present two algorithms. One of these algorithms solves the problem under assumption that optimum partitioning contains at least one unaffected vertex, the other one solves the problem under assumption that 
all vertices are affected in optimum partitioning. In order to estimate running time of the algorithms we use the following lemma.

\begin{lemma}\cite{FominKPPV14}\label{bound}
For any non-negative integer $a$, $b$ we have
$
{a + b \choose b} \leq 2^{2\sqrt{ab}}.
$
\end{lemma}

At first, consider a case when there is at least one unaffected vertex in optimum partitioning.
\begin{lemma}\label{alg32}
 Let  $G$ be a connected graph with diameter at most $2$. 
 If there is an optimum partitioning  $C_1 \sqcup C_2\sqcup\dots \sqcup C_\ell$ of $G$  
 into highly connected graphs such that set of unaffected vertices is not empty then 
 \hcd can be solved in  $\cOs(2^{\frac{3k}{2}})$ time.
\end{lemma}

\begin{proof}
Let us fix some unaffected vertex $u$ (in algorithm we simply brute-force all $n$ possible values for unaffected vertex $u$). By Lemma~\ref{cluster} highly connected  graph 
$C(u)$ contains at least $n - k$ vertices. 
As $u$ is unaffected then $N(u) \subset C(u)$ and $|N(u)| > \frac{|C(u)|}{2}$. 
Consider set $V \setminus N[u]$. And partition it into two subsets 
$W_{1,2} \sqcup W_{\geq 3}$, where $W_{1,2} = \{v | 1 \leq |N(u) \cap N(v)| \leq 2\}$, and
$W_{\geq 3} = \{v | 3 \leq |N(u) \cap N(v)| \}$. From  lemma~\ref{neigh} follows that $W_{1,2} \cap C(u) = \emptyset$.  Note that knowing set $C_{part}=C(u)\cap W_{\geq 3}$ 
we can find set $C(u)=C_{part}\cup N[u]$ and after this simply run algorithm from Theorem~\ref{exact:hcd} on set $V(G)\setminus C(u)$. We implement this approach.

We know that $N[u]\sqcup C_{part}=C(u)$ and $C(u) \leq 2k$. As $|C_{part}|\leq \frac{C(u)}{2}$ it follows that $|C_{part}|\leq k$. Brute-force over all possible values of  $s=|C_{part}|$. Having fixed value of $s$ we enumerate all subsets of $W_{\geq3}$ of size $s$. All such subsets are potential candidates for a $C_{part}$ role. It is possible to enumerate candidates with polynomial delay i.e. in  $O^*({{ |W_{\geq3}| } \choose {|C_{part}|}})$ time.

For each listed candidate we run algorithm from Theorem~\ref{exact:hcd}. Let $R=W_{\geq 3} \setminus C_{part}$. Hence, the overall running time for a fixed $|C_{part}|$ is bounded by 
$\cOs(2^{|R \cup W_{1,2}|})  {{ |W_{\geq3}| } \choose {|C_{part}|}} = \cOs(2^{|R \cup W_{1,2}|})  {{ |C_{part}|+|R| } \choose {|C_{part}|}}$. 
By Lemma~\ref{bound}  we have:
  
\(
	\cOs(2^{|R \cup W_{1,2}|})  {{ |C_{part}|+|R| } \choose {|C_{part}|}} = 
	\cOs(2^{{2 \sqrt{|C_{part}||R|}} + |R| + |W_{1,2}|}).
\)
  
We know that $|C_{part}| \leq k$, $3|R| + |W_{1,2}| \leq k$, hence
\( 
	\cOs(2^{{2 \sqrt{|C_{part}||R|}} + |R| + |W_{1,2}|}) \leq \cOs(2^{2 \sqrt{k|R|} - 2|R| + k}).
\)
The function $g(t)=2 \sqrt{kt} - 2t + k$ attains it maximum when $t=\frac{k}{4}$. 
So the running time in the worst case is $\cOs(2^{1.5k})$.
\end{proof}
The following Algorithm~\ref{alg:unaffected} illustrates the proof of last Lemma.
\begin{algorithm}\caption{}\label{alg:unaffected}
	\begin{algorithmic}
		\Function{$UNAFFECTED$}{$G = (V, E), k$}
		\For{$u \in V$}
		\State $W_{1,2} = \{v | v \in V \setminus N[u], |N(v) \cap N(u)| \leq 2\}$
		\State $W_{\geq 3} = \{v | v \in V \setminus N[s], |N(v) \cap N(u)| \geq 3\}$
		\For{$s : s < |N(u)| \And s \leq k \And 3(|W_{\geq 3}| - s)+|W_{1,2}| \leq k$}
		\For{$C_{part} \subseteq W_{\geq 3} \And |C_{part}| = s$}
		\State $Q = N[u] \cup C_{part}$
		\If{$G[Q]$ is highly connected}
		\If{$EXACT(G[V \setminus Q], k - |E(Q, V \setminus Q)|)$}
		\State \Return YES
		\EndIf
		\EndIf
		\EndFor
		\EndFor
		\EndFor
		\State \Return NO
		\EndFunction
	\end{algorithmic}
\end{algorithm}

It is left to construct an algorithm for a case in which all vertices are affected in optimum partitioning.
 First of all note that 
if  $n \leq 1.57k \leq k\log_{2}{3}$ we can simply run Algorithm~\ref{exact:hcd} and it finds an answer in 
$\cOs(2^n) = \cOs(3^k)$ time.  Taking into account that all vertices are affected we have that 
$n\leq 2k$. So we may assume that  $1.57 k \leq n \leq 2k$.

\begin{lemma}\label{big}
Let $G$ be a graph with diameter $2$ and $|V(G)| \geq 1.57{k}$. Moreover, $(G, k)$  \hcd problem  admits correct partitioning  into highly connected components $C_1\sqcup C_2\sqcup \dots \sqcup C_\ell $  such that all vertices are affected in this partitioning. Then there are two highly connected components 
$C_i, C_j$  such that $|C_i| + |C_j| \geq {n - k}$.
\end{lemma}

\begin{proof}
Let $E'$ be set of deleted edges for partitioning  $C_1\sqcup C_2\sqcup \dots \sqcup C_\ell $. From $n \geq 1.57 k$ follows that in graph $(V(G), E')$ 
there is a vertex $s$ of degree $1$,  let $st \in  E'$ be the edge. 
We prove that $C(s), C(t)$ are desired highly connected components. 
As $diam(G)\leq 2$ then for any vertex $v \in V(G)\setminus C(s) \setminus C(t)$ 
there is path of length at most $2$ from $s$ to $v$. 
Hence, any vertex  $v \in V(G)\setminus C(s) \setminus C(t)$  should be connected with $C(s) \cup C(t)$ in graph $G$.
As $|E'|\leq k$ then $V(G) \setminus (C(s)\cup C(t)) \leq k$. So $|C(s)|+|C(t)|\geq n-k$.
\end{proof}

Now we brute-force all vertices as candidates for a role of vertex $s$, i.e. vertex of degree $1$ 
in solution $E'$. Consider two possibilities either $|C(s)| > 2n - 3.14k$ or $|C(s)| \leq 2n - 3.14k$.

Consider the first case, if  $|C(s)| > 2n - 3.14k$, then we find solution in 
$\cOs(2^{n - \frac{|C(s)|}{2}}) = \cOs(3^{k})$ time. In order to do this we consider $deg_G(s)$ cases. Each case correspond to a different edge $st$ incident with $s$.  Such an edge we treat as the only edge incident with $s$ from $E'$. Having fixed an edge $st$ being from $E'$ we know that all other edges incident with $s$ belong to $E(C(s))$.
Denote the set of endpoints of these edges to be $U$. So we can identify at least 
$\frac{|C(s)|}{2}$ vertices from $C(s)$. 
Now we can apply the same technique as in proof of Theorem~\ref{exact:hcd}.

We define three functions $f, g, h$ over subsets of $W = V \setminus U$. 
\begin{itemize}
\item $f(S) = |E(S, W \setminus S)|$ if $G[S]$ is highly connected, otherwise it is equal to $\infty$. 
\item $h(S) = \min\limits_{i}(f^{\ast i}(S))$.
\item $g(S) = 2|E(W \setminus S, U)| + |E(S, W \setminus S)|$ if $G[U \cup S]$ is highly connected otherwise it is $\infty$.
\end{itemize}

Let us provide some intuition standing behind the formulas. Value $f(S)$ indicate number of vertices 
that we have to delete in order to separate highly connected graph $G[S]$. 
$h(S)$ is a number of edges needed to be deleted in order to separate $G[S]$ into highly connected components. 
$g(S)$ in some sense is a number of edge deletion needed to create a highly connected 
component $U\cup S$ which contains vertex $s$. We show that to solve the problem it is enough to compute 
$(g \ast h)(W)$. In similar way to Theorem~\ref{exact:hcd}  
$(g \ast h)(W)/2$ equals to a number of optimum edge deletions. 
Note that all deleted edges not having endpoints in $C(s)$ will be calculated two times, one for each of its incident  highly connected component, see definition of function $h$. Each edge of $E'$ having an endpoint in  $U$ is counted twice in first term of function $g$. And finally each edge from $E'$ having endpoint 
in $C(s)\setminus U$ is counted twice, once in second term of the formula of $g$, 
and once in the formula of $h$. So $(g \ast h)(W)/2$ is required number of edge deletions.

Second case, if $|C(s)| \leq 2n - 3.14k$ then
$n - k  \leq |C(s)| + |C(t)| \leq 2n - 3.14k + |C(t)| .$

It follows that $|C(t)| + 2n - 3.14k  \geq n - k$. Hence,  $C(t) \geq 2.14k - n \geq 0.14k$.
It means that in $C(t)$ there is a vertex of degree at most $7$ in graph $(V(G),E')$. 
We brute-force all candidates for such vertex and for such edges  from $E'$. 
Having fixed the candidates, vertex $t'$ and at most seven edges, we  identify  more than 
a half vertices from $C(t')=C(t)$ in the following way. All edges incident to $t'$ except just fixed set of candidates  
belong to $C(t)$. 
Denote the endpoints of these edges as $U_t$. In the same way, all edges incident with  $s$ except $st$ 
belong to $C(s)$. Denote by $U_s$ endpoints of edges incident with $s$ except 
the  edge $st \in E'$. 
Let $U = U_s \cup U_t$. Below we show how to solve obtained problem  
in $\cOs\left(2^{n - \frac{1}{2}{\left(|C(s)| + |C(t)|\right)}}\right)$ time. As in previous case we apply idea similar to 
algorithm from Theorem~\ref{exact:hcd}. Now we present only functions which convolution give an answer.
As the further details are identical to Theorem~\ref{exact:hcd}.

Our functions are defined over subsets of a set  $W = V \setminus U$.
\begin{itemize}
\item $f(S) = |E(S, W \setminus S)|$ if $G[S]$ is highly connected, otherwise $\infty$. 
\item $h(S) = \min\limits_{i}\left(f^{\ast i}(S)\right)$.
\item $g_s(S) = 2|E(S, U_t)| + |E(S, W \setminus S)|$ if $G[S \cup U_s]$ is highly connected, otherwise  $\infty$.
\item $g_t(S) = 2|E(S, U_s)| + |E(S, W \setminus S)|$ if $G[S \cup U_t]$ is highly connected, otherwise $\infty$.
\end{itemize}
 
The only difference from previous case is that we constructed  two functions  $g_s, g_t$ instead of just one function $g$
as now we know two halves of two guessed highly connected components.
Minimum number of edge deletions in \texttt{YES}-instance 
separating clusters $C(s), C(t)$ ($U_s\subseteq C(s), U_t\subseteq C(t)$) 
is  $(h \ast g_s \ast g_t) (W)/2$.  
So in this case we need $\cOs(2^{|W|})$  running time 
which is $\cOs\left(2^{n -  \frac{(n- k)}{2}}\right) = \cOs\left(2^{\frac{3k}{2}}\right)$. 

\end{proof}
Pseudo-code for algorithm from previous lemma is shown in Algorithm~\ref{alg:affected}.

\begin{algorithm}\caption{}\label{alg:affected}
	\begin{algorithmic}
		\Function{$AFFECTED$}{$(V, E), k$}
		\If{$|V| \leq 1.57k$}
		\State \Return $EXACT((V, E),k)$
		\EndIf
		\If{$|V|>2k$}                                       
		\State \Return NO
		\EndIf
		\For{$st \in E$}
		\State $U(s) = N[s] \setminus \{t\}$
		\If{$|U(s)| > n-1.57k$}
		\State Compute $f, h, g, g \ast h$ for all subsets of $V \setminus U(s)$
		\If{$(g \ast h)(V \setminus U(s)) \leq 2k$}
		\State \Return YES
		\EndIf
		\Else
		\For{$ 0 \leq l \leq 7, (t'y_1, \dotsc ,t'y_l) \in E^{l}$}
		\State $U(t') = N[t'] \setminus \{y_1, \dotsc, y_l\}$
		\State $U = U(s) \cup U(t)$                                                                                                                                   
		\If{$U(s) \cap U(t') = \emptyset \land |U| \geq \frac{n-k}{2} $ }
		\State Compute $f, h, g_s, g_t, h \ast g_s \ast g_t$ for  all subsets of $V \setminus U$
		\If{$(h {\ast} g_s \ast g_t)(V \setminus U) \leq 2(k - |E(U(s), U(t'))|$}
		\State \Return YES                                                                                                                                            
		\EndIf
		\EndIf
		\EndFor
		\EndIf
		\EndFor
		\State \Return NO
		\EndFunction
	\end{algorithmic}
\end{algorithm}

\subsection{\phcd}

\defproblemu{\phcd}
{
	Graph \(G = (V, E)\), integer numbers \(p\) and \(k\).
}
{
	Is there a subset of edges \(E' \subset E\) of size at most \(k\) such that \(G - E'\)
	contains at most \(p\) connected components and each component is highly connected?
}

Our algorithm for \phcd is insipired by algorithm for  \textsc{$p$-Cluster Editing} by Fomin et al.~\cite{FominKPPV14}.

First of all, we prove  an upper bound on the number of small cuts in highly connected graph.

\begin{lemma}\label{cuts}
	Let \(G=(V, E)\) be highly connected graph, 
	\(X = \argmin\limits_{\substack{S \subset V \\ \frac{|V|}{4} \leq |S| \leq 
			\frac{3|V|}{4}}}|E(S, V \setminus S)|\), and \(Y = V \setminus X\), then
	\begin{enumerate}[i)]
		\item 
		If \(|E(X, Y)| \geq \frac{|V|^2}{100}\) then for any partition of  \(V = A\sqcup B\) we have  \(|E(A,B)| \geq \frac{|A| \cdot |B|}{100}\,.\)
		\item 
		If \(|E(X, Y)| < \frac{|V|^2}{100}\) then for any partition of \(V = A\sqcup B\) we have: \\
		\( |E(A \cap X, B \cap X)| \geq  \frac{|X \cap A|\cdot|X \cap B|}{100},\)
		\(|E(A \cap Y, B\ \cap Y)|\geq \frac{|Y \cap A|\cdot|Y \cap B|}{100},\)\\
		\( |E(A, B)| \geq    \frac{|X \cap A|\cdot|X \cap B|}{100} + \frac{|Y \cap A|\cdot|Y \cap B|}{100}.\)
	\end{enumerate}
\end{lemma}
\begin{proof}
	
	$i)$ Let \(V = A\sqcup B\). Without loss of generality  \(|A| < |B|\).
	
	If \( \frac{|V|}{4}\leq |A|\) then \(|E(X, Y)| \leq |E(A, B)|\). 
	Hence, \(|E(A, B)| \geq |E(X, Y)| \geq \frac{|V|^{2}}{100} \geq \frac{|A| \cdot |B|}{100}\,.\)
	
	If \(|A| < \frac{|V|}{4}\)  then \(|E(A, B)| \geq \sum\limits_{v \in A}(\deg(v) - |A|)\). 
	As \(\deg(v) > \frac{|V|}{2}\) for all $v \in V(G)$, we have  
	\(|E(A, B)| \geq |A|\left(\frac{|V|}{2} - |A|\right) \geq \frac{|A|\cdot|V|}{4} \geq 
	\frac{|A|\cdot|B|}{4} \geq \frac{|A|\cdot|B|}{100}\,.\)
	
	$ii)$ Note that \(|E(A, B)| \geq |E(A \cap X, B \cap X)| + |E(A \cap Y, B \cap Y)|\,.\) So it is enough to  prove that \(|E(A \cap X, B \cap X)| \geq \frac{|A \cap X|\cdot|B \cap X|}{50}\), as the proof of
	\(|E(A \cap Y, B \cap Y)| \geq \frac{|A \cap Y|\cdot|B \cap Y|}{50}\) is analogous. 
	The sum of these two inequalities gives the proof of the theorem.
	
	Without loss of generality \(|B \cap X| \leq |A \cap X|\). Hence, \(\frac{|V|}{8} \leq |A \cap X|\) and \(|B \cap X| \leq \frac{3|V|}{8}\,.\) Consider two cases: \(|A \cap X| \geq \frac{|V|}{4}\) and \(|A \cap X| < \frac{|V|}{4}\,.\)
	
	Consider case when \(|A \cap X| \geq \frac{|V|}{4}\).  At first we prove \(|E(A \cap X, B \cap X)| \geq |E(B \cap X, Y)|\). 
	It is known that:
	\begin{equation}\label{cut}
	|E(A \cap X, V \setminus \left( A \cap X\right))| = |E(X, Y)| - |E(B \cap X, Y)| + |E(A \cap X, B \cap X)|\,,
	\end{equation} 
	\(|A \cap X| \geq \frac{|V|}{4}\), and \(|V \setminus \left(A \cap X\right)| \geq |Y| \geq \frac{|V|}{4}\),
	it means \(|E(A \cap X, V \setminus \left(A \cap X\right))| \geq |E(X, Y)|\). 
	The last inequality  and~(\ref{cut}) imply
	\(|E(A \cap X, B \cap X)| \geq |E(B \cap X, Y)|\). 
	It follows that \( 2|E(A \cap X, B \cap X)| \geq   |E(B \cap X, A \cap X)| + |E(B \cap X, Y)| =|E(B \cap X, V \setminus \left(B \cap X\right)|  \).  
	
	As \( \frac{3|V|}{8} \geq |B \cap X| \)  and \(  |E(B \cap X, V \setminus \left(B \cap X\right) ) | \geq |B \cap X| \left(\frac{|V|}{2} - |B \cap X|\right)  \)  we have  \( |E(B \cap X, V \setminus \left(B \cap X\right))| \geq \frac{|B \cap X| \cdot |V|}{8}\). Hence, \( |E(A \cap X, B \cap X)| \geq \frac{|B \cap X| \cdot |V|}{16} \geq \frac{|B \cap X| \cdot |V|}{100} \).

	It is left to consider case \(|A \cap X| < \frac{|V|}{4}\). Note that 
	\(|E(A \cap X, B \cap X)| = |E(A \cap X, V \setminus (A \cap X))| - |E(A \cap X, Y)|\).
	As \(  \frac{|V|}{4} > |A \cap X|\) we have
	\(|E(A \cap X, V \setminus (A \cap X))| \geq 
	|A \cap X|\left(\frac{|V|}{2} - |A \cap X|\right) \geq \frac{|V|}{8}\cdot\frac{|V|}{4} \geq \frac{|V|^2}{32}\). 
	We know that \(|E(A \cap X, Y)| \leq |E(X, Y)| \leq \frac{|V|^2}{100}\), hence
	\(|E(A \cap X, B \cap X) \geq \frac{|V|^2}{32} - \frac{|V|^2}{100} > \frac{|V|^2}{50} \geq 
	\frac{|A \cap X| \cdot |B \cap X|}{100}\,.\)
\end{proof}

\begin{definition}
	A partition of \(V = V_1 \sqcup V_2\) is called a \(k\)-cut of \(G\) if \(|E(V_1, V_2)| \leq k\,.\)
\end{definition}

The following lemma limits number of $k$-cuts in a disjoint union of highly connected graphs.

\begin{lemma}\label{boundcut}
	If \(G = (V, E)\) is  a union of $p$ disjoint highly connected components and \(p \leq k\) then 
	the number of \(k\)-cuts in \(G\) is bounded by \(2^{\cO\left(\sqrt{pk}\right)}\,.\)
\end{lemma}
\begin{proof}
	
	Let $G$ be a disjoint union of highly connected components \(C_1, \dotsc, C_p\). 
	For each $C_i$ we consider sets $X_i,Y_i$ where $E(X_i,Y_i)$
	is a minimum cut of $C_i$ and $C_i=X_i \sqcup Y_i$. We construct a new partition  \(C'_1, \dotsc, C'_q\) of $V(G)$.
	The new partition is obtained from partition \(C_1 \sqcup \dotsc \sqcup C_p\)  in the following way:
	if $|E(X_i,Y_i)| <  |C_i^2|/100$ then we split $C_i$ into two sets $X_i, Y_i$ otherwise we take $C_i$ 
	without splitting. Note that \(p \leq q \leq 2p\) as we either split $C_i$ into to parts or leave it as is.

	We bound number of $k$-cuts of graph $G$ in two steps. In first step we bound number of cuts $V_1,V_2$ 
	such that  \( |V_1 \cap C'_i|=x_i\) and 
	\(|V_2 \cap C_i'| = y_i\) where $x_i, y_i$ are some fixed integers. In second step we bound number of 
	tuples $(x_1,\dots,x_q, y_1, \dots, y_q)$ for which there is at least one 
	$k$-cut $V_1,V_2$ satisfying conditions  \( |V_1 \cap C'_i|=x_i\),  \( |V_2 \cap C'_i| = y_i\).


	If \(x_i, y_i\) are fixed and $x_i+y_i = |C'_i|$  the number of partitions of $C'_i$ is equal to \(\binom{x_i+y_i}{x_i}\).  Note that  by Lemma~\ref{bound} we have \(\binom{x_i + y_i}{x_i} \leq 2^{\sqrt{x_iy_i}}\).
	Observe that there are at least \(\frac{x_iy_i}{100}\) edges 
	between \(V_1 \cap C'_i\) and \(V_2 \cap C'_i\) by Lemma~\ref{cuts}. 
	So if \(V_1 \sqcup V_2\) is partition of \(V\) then \(\sum\limits_{i = 1}^{q}x_iy_i \leq 100k\). 
	Applying  Cauchy–Schwarz
	inequality we infer that \(\sum\limits_{i=1}^{q}\sqrt{x_iy_i} \leq 
	\sqrt{q}\cdot\sqrt{\sum_{i=1}^{q}x_iy_i} \leq \sqrt{200pk}\). 
	Therefore, the number of considered cuts is at most
	\(
	\prod\limits_{i=1}^{q}\binom{x_i+y_i}{x_i} \leq 2^{2\sum_{i=1}^{q}\sqrt{x_iy_i}} \leq 2^{\sqrt{800pk}}.
	\)
	
	Now we show bound for a second step i.e. number of possible tuples $(x_1,\dots,x_q, y_1, \dots, y_q)$ generating at least one $k$-cut. Note that \(\min\{x_i, y_i\} \leq \sqrt{x_iy_i}\). Hence, 
	\(\sum\limits_{i=1}^{q}\min(x_i, y_i) \leq \sqrt{100qk}\,.\) Tuple $(x_1,\dots,x_q, y_1, \dots, y_q)$ can be generated in the following way: at first we choose which value is smaller $x_i$ or $y_i$. Then we express
	\(\sqrt{\left\lfloor 100qk\right\rfloor}\) as 
	a sum of \(q+1\) non-negative numbers: \(\min\{x_i, y_i\}\) for \(1 \leq i \leq q\) and the rest
	\(\sqrt{\left\lfloor 100qk\right\rfloor} - \sum\limits_{i=1}^{q}\min(x_i, y_i)\). 
	
	The number of choices in  the first step of generation is equal to \(2^{q} \leq 2^{\sqrt{2qk}}\), and number of ways to expreess $\sqrt{100qk}$ as a sum of $q+1$ number is at most
	\(\binom{\sqrt{100qk}+q + 1}{q} \leq 2^{\sqrt{100qk} + q+1} \leq 2^{\sqrt{100qk} + \sqrt{2qk}+1}\). 
	Therefore, the total number of partitions is bounded by \(2^{c\sqrt{pk}}\) for some constant \(c\).
\end{proof} 

The last ingredient  for our algorithm is the following lemma proved by Fomin et al.\cite{FominKPPV14}

\begin{lemma}\cite{FominKPPV14}\label{enumerate:cuts}
	All cuts  \((V_1, V_2)\) such that \(|E(V_1, V_2)| \leq k\) of a graph \(G\) can be enumerated with polynomial time delay.
\end{lemma}

Now we are ready to present a final theorem.

\begin{theorem}
	There is a \(\cOs(2^{\cO(\sqrt{pk})})\) time algorithm for \phcd problem.
\end{theorem}

\begin{proof}
	First of all we solve the problem in case of connected graph. Denote by \(\mathcal{N}\) set of  all \(k\)-cuts in graph $G$. All elements of set \(\mathcal{N}\) can be enumerated with a polynomial time delay. If \(G\) is a union of \(p\) clusters 
	plus some edges then the size of \(\mathcal{N}\) is bounded by \(2^{c\sqrt{pk}}\) by Lemma~\ref{boundcut} (as additional edges only decrease number of $k$-cuts).
	Thus, we enumerate \(\mathcal{N}\) in time \(\cOs(2^{\cO(\sqrt{pk})})\). 
	If we exceed the bound \(2^{c\sqrt{pk}}\) given by Lemma~\ref{boundcut} we know that we can 
	terminate our algorithm and return answer \texttt{NO}. So we may assume that  we enumerate the whole \(\mathcal{N}\) and it contains at most \(2^{c\sqrt{pk}}\) elements.  
	
	We construct a directed graph \(D\), whose vertices are elements of a set
	\(\mathcal{N}\times\{0, 1, \dotsc, p\}\times\{0, 1, \dotsc, k\}\), 
	note that \(|V(D)| = 2^{\cO(\sqrt{pk})}\).
	We add arcs going from \(((V_1, V_2), j, l)\) to \(((V'_1, V'_2), j + 1, l')\), 
	where \(V_1 \subset V'_1\), \(G[V'_1 \setminus V_1]\) is highly connected graph,
	\(j \in \{0, 1, \dotsc, p-1\}\), and \(l' = l + |E(V_1, V'_1 \setminus V_1)|\).
	The arcs can be constructed in \(2^{\cO(\sqrt{pk})}\) time. 
	We claim that the answer for an instance \((G, p, k)\) is equivalent to existence of path 
	from a vertex \(((V, \emptyset), 0, 0)\) to a vertex \(((\emptyset, V), p', k')\) for some $p'\leq p, k'\leq k$.
	
	In one direction, if there is a path from \(((\emptyset, V),0,0)\) to \(((V,\emptyset),p',k')\) 
	for some \(k' \leq k\) and \(p'\leq p\), then the consecutive sets \(V'_1 \setminus V_1\) along the path 
	form highly connected components. Moreover,  number of deleted edges from \(G\) is equal to last coordinate which is smaller than \(k\). 
	
	Let us prove the opposite direction. Let assume that we can delete at most $k$ edges and get a graph with
	highly connected components \(C_1, \dotsc, C_p\). Let us denote \(T_i = \cup_{j < i}V(C_i)\), 
	\(l_{i + 1} = l_{i} + |E(T_{i + 1} \setminus T_{i}, T_{i})|\) then 
	the vertices $((T_i, V \setminus T_i), i - 1, l_i)$ constitute desired path in graph \(D\). 
	
	Reachability in a graph can be tested in a linear time with respect to the number of vertices and arcs.
	To concude the algorithm we simply test the reachability in the graph \(D\).
	
	It is left co consider a case when $G$ is not connected. Let assume that \(G\) consist of \(q\) connected components \(C_1, \dotsc, C_q\) then  for each connected component \(C_i\) we find all \(p' \leq p\) and $k' \leq k$  such that \((C_i, p', k')\) is \texttt{YES}-instance. After this we construct auxiliary directed graph \(Q\) with a set of vertices 
	\(\{0, \dotsc, q\} \times \{0, \dotsc, p\} \times \{0, \dotsc, k\}\). 
	We add arcs going from 
	\((i, a, b)\) to \((i + 1, a + p', b + k')\) if \((C_i, p', k')\) is a \texttt{YES}-instance. 
	Using similar arguments as before it could be shown that reachability of vertex $(q,p',k')$ from vertex $(0,0,0)$ is equivalent to possibility delete $k'$ edges and get $p'$ highly connected components. 
\end{proof}

\section{Algorithms for finding a subgraph}

\subsection{\seed}

\defproblemu{\seed}{
Graph $G = (V, E)$, subset  $S \subseteq V$ and integer numbers  $a$ and $k$.
}{
Is there a subset of edges $E' \subseteq E$ of size at most $k$  such that $G - E'$ contains only isolated vertices and one highly connected component $C$ with $S \subseteq V(C)$ and $|V(C)| = |S| + a$.
}

H{\"u}ffner et al.~\cite{HuffnerKS15} constructed an algorithm with running time
$\cO(16^{k^{0.75}} + k^2nm)$ for \seed problem. We improve the result to \\ $\cOs\left(2^{\cO(\sqrt{k}\log{k})}\right)$ time algorithm.

\begin{theorem}\label{seed:main}
There is $\cOs(2^{\cO(\sqrt{k}\log{k})})$ time algorithm for \seed problem.
\end{theorem}

We rely on the following theorem proved in \cite{HuffnerKS15}.

\begin{theorem}\label{seed:kernel}\cite{HuffnerKS15}
	Any instance of \seed problem can be  transformed in $\cO(k^2nm)$ time into equivalent 
	instance with at most  $2k+\frac{4k}{a}$ vertices and at most ${2k \choose 2} + k$ edges.
\end{theorem}
\begin{proof}[Proof of theorem~\ref{seed:main}]
	By Theorem~\ref{seed:kernel} we construct an equivalent instance with at most $2k+\frac{4k}{a}$
	vertices and at most  ${2k \choose 2} + k$ edges. We consider two cases  $a \leq 2\sqrt{k}$ and $a > 2\sqrt{k}$.
	
	\textbf{Case 1: $a \leq 2\sqrt{k}$.} 
	
	In order to solve the problem we simply brute-force over all possible candidates. 
	We consider all vertex subsets $V'$ of size at most $2 \sqrt{k}$ and  in each branch 
	check whether $S\cup V'$ is an answer. 
	It is easy to see that the algorithm is correct. Up to polynomial factor the running time of such algorithm 
	is equal to number of candidates $V'$. Hence, the running time is at most 
	$O^*\left({ 2k + \frac{4k}{a} \choose a}\right) \leq { 6k \choose a} \leq (6k)^{a} \leq 2^{\cO(\sqrt{k}\log{k})}$.

	\textbf{Case 2:  $a > 2\sqrt{k}$.} 
	
	Since $a > 2\sqrt{k}$ then the size of highly connected component from the solution is 
	at least  $2\sqrt{k}$. So, if  $deg(w) <\sqrt{k}$ then $w$ does not belong to the highly connected component 
	from solution. In this case  we  delete vertex $w$ and all its edges, decreasing parameter $k$ by $deg(w)$. 
	Hence, we can assume that degree of all vertices is at least $\sqrt{k}$. 
	However, in such case  at most $2\sqrt{k}$ vertices are not present in highly connected component of the solution. As otherwise we have to delete more than  $2 \sqrt{k} \cdot \sqrt{k}$ edges. So now, we simply brute-force all subsets of vertices $F$ 
	that are no part of a highly connected graph. In order to do this we have to consider at most  
	$\cOs\left(\sum_{i \leq 2\sqrt{k}}{ n \choose i}\right) = \cOs\left({6k \choose 2\sqrt{k}}\right) 
	= \cOs\left(2^{\cO(\sqrt{k}\log{k})}\right)$ cases.

	So the running time for \textbf{Case 2} match with the running time of case \textbf{Case 1}.  Hence, the running time of the whole algorithm is $\cOs(2^{\cO(\sqrt{k}\log{k})})$.
\end{proof}

\subsection{\hcs}

\defproblemu{\hcs}{Graph $G = (V, E)$, integer $k$, integer $s$.}{
Is there a set of vertices $S$ such that $|S|=s$, $G[S]$ is highly connected graph 
and $|E(S, V \setminus S)| \leq k$.
}


 H{\"u}ffner et al.~\cite{HuffnerKS15} proposed \( \cOs(4^k) \) algorithm for \hcs problem, in this work we construct subexponential algorithm for the same problem with running time \( \cOs(k^{\cO(k^{2/3})}) \).

In order to solve \hcs problem  H{\"u}ffner et al. in~\cite{HuffnerKS15} constructed algorithm for a more general problem: 

\defproblemu{\fhcs}{Graph $G = (V, E)$, integer $k$, integer $s$, function \( f : V \to \mathbb{N} \).}{
Is there  a set of vertices $S$ such that  $|S|=s$, $G[S]$ is highly connected and 
$|E(S, V \setminus S)|  + \sum\limits_{v \in S}f(v)\leq k$.}

Our algorithm uses reduction rules proposed in~\cite{HuffnerKS15}. 
Here, we state the reduction rules without proof, as the proofs can be found in \cite{HuffnerKS15}.

\begin{reduce}\label{hcs:smallcomp}
If $G$ contains connected component $C$ of size smaller than $s$ then delete $C$ i.e. solve instance
$(G \setminus C, f, k)$.
\end{reduce}

\begin{reduce}\label{hcs:bigcut}
Let $G$ contains connected component $C = (V', E')$ with minimal cut bigger than $k$. 
If $C$ is highly connected graph,  $|V'|=s$ and $ \sum\limits_{s \in V'}f(s) \leq k$ 
then output a trivial \texttt{YES}-instance otherwise remove $C$, 
i.e. consider instance  $(G \setminus C, f, k)$ of \fhcs problem.
\end{reduce}

\begin{reduce}\label{hcs:cut}
Let $G$ contains connected component  $C$ with minimal cut  $(A, B)$ of size at most $\frac{s}{2}$.
We define function $f'$ in the following way: for each vertex $v \in A$ 
$f'(v):= f(v) + |N(v) \cap B|$ and for each  $v \in B$ we let $f'(v) := f(v)  + |N(v) \cap A|$. 
Replace original instance with an instance $(G\setminus~{E(A, B)}, f', k)$.
\end{reduce}
\begin{lemma}\label{hcs:bound}
Rules \ref{hcs:smallcomp}, \ref{hcs:bigcut}, \ref{hcs:cut}  can be exhaustively applied in time $\cO((sn + k)m)$.  If rules \ref{hcs:smallcomp}, \ref{hcs:bigcut}, \ref{hcs:cut} are not applicable then $k > \frac{s}{2}$.
\end{lemma}


We also  use following  Fomin and Villanger's result.
\begin{proposition}\cite{FominV12}\label{conenum}  For each vertex  $v$ in graph $G$ and integers $b, f \geq 0$ 
number of connected induced subgraphs
$B \subseteq V(G)$ satisfying the following properties 
$v \in B$, $|B| = b + 1$, $|N(B)| = f$;
is at most ${b + f}\choose{b}$. Moreover, all these sets can be enumerated in time 
\(\cO\left({{b + f}\choose{b}}(n + m)b(b+f)\right).\,\)
\end{proposition}

Now we have all ingredients for out algorithm. 

\begin{theorem}\label{hcs:alg}
\fhcs can be solved in time   $2^{\cO(k^{{2/3}}\log{k})}$. 
\end{theorem}
\begin{proof}
First of all we exhaustively apply reduction rules \ref{hcs:smallcomp}, \ref{hcs:bigcut}, \ref{hcs:cut}. 
From Lemma~\ref{hcs:bound} follows that we may assume
$2k > s$. We consider two cases either $k^{{2/3}} < s$ or $k^{{2/3}} \geq s$.

\textbf{Case 1:}  $s \leq k^{{2/3}}$. Enumerate all induced connected subgraphs 
$G'=(V',E')$ such that $|V'|=s$ and $N(V')\leq k$. 
If desired $S$ exists than it is among enumerated sets. 
From Proposition~\ref{conenum} follows that number of such sets is at most  $nk\cOs({s + k \choose s})$.  
As  $s < 2k$ and 
\(
s < k^{2/3}$ we have $nk\cOs({s + k \choose s})\leq \cOs((s+k)^{s}) \leq \cOs(2^{k^{{2/3}}\log{k}})
\). 
Hence, in time $\cOs(2^{k^{{2/3}}\log{k}})$ we can enumerate all potential candidates $S'$.  
For each candidate we check in polynomial time  whether 
$G[S']$  is highly connected and $|E(S',V{\setminus}S')| + \sum\limits_{v \in S'}f(v) \leq k$.
 
\textbf{Case 2:} $k^{\frac{2}{3}} < s$. Let set  $S$ be a solution. Define edge set 
$E' = E(S, V \setminus S)$. Consider function $d : S \to \mathbb{N}$ 
where $d(v)=|N(v)\cap(V{\setminus}S)|$. 
As $\sum\limits_{v \in S}d(v) = |E(S,V{\setminus}S)| \leq k$ then there is a vertex $v \in S$ such 
that $d(v) \leq \frac{k}{s} < k^{\frac{1}{3}}$. Note that for such  $v$  we have 
 $|N(v)|=|N(v)\cap S|+|N(v)\setminus S|\leq s+k^{\frac{1}{3}} $. 
We branch on possible values of such vertex and a set of its neighbors that do not belong to $S$. 
In order to do this we have to  consider at most  
\(
  n\sum\limits_{i \leq k^{{1/3}}}{s + k^{{1/3}} \choose i} \leq 
  {n}k^{{1/3}}2^{2\sqrt{(s + k^{1/3}-i)i}} \leq {n}k^{{1/3}}2^{2\sqrt{3k^{{4/3}}}}
   = n2^{O(k^{2/3})} \, 
\)
  cases. Knowing vertex $v \in S$ and $N(v)\setminus S$ we  find $N(v) \cap S$. So we already identified at least  $\frac{s}{2}+1$ vertices from $S$, let denote this set by $W$. Now we start branching procedure that in right branch extend set $W$ into a solution set $S$. Branching procedure takes as an input tuple  $(G, k, s', W, B)$ where $W$ is a set of vertices determined to be in solution $S$, $B$ is a set of vertices determined  to be not in solution, $k$ number of allowed edge deletions, $s'=s-|W|$ number of vertices that is left to add. The procedure pick a vertex $w \notin W\cup B$ and consider two cases either $w\in S, w\notin B$ or $w \notin S, w \in B$. 
  The first call of the procedure is performed on tuple $(G, k-|E(W,N(v)\setminus W)|, s-|W|, W, \emptyset)$.

Consider arbitrary vertex $x\in V\setminus (W \cup B)$. 
If $x \in S$ then $|N(x)\cap S| \geq  \frac{s}{2}$. 
Hence, $|N(x) \cap W| \geq |N(x)\cap S|- |S\setminus W| \geq \frac{s}{2} - (s-|W|) =|W|-\frac{s}{2} $. 
So any vertex $x$ such that $|N(x)\cap W|<|W|-\frac{s}{2}$ cannot belong to solution $S$ and we safely put $x$ to $B$. 
Otherwise, we run our procedure  on tuples $(G, k - |N(x) \cap B|, s'- 1, W \cup x, B)$ and 
$(G, k - |N(x) \cap W|, s' , W , B \cup x)$. 
Note that we stop computation in a branch if $k'\leq 0$ or $s' = 0$.  
It is easy to see that the algorithm is correct.

It is left to determine the running time of the algorithm.
Note that procedure contains two parameters $k$ and $s'$. 
In one branch we decrease value of $s'$ by one in the other branch we decrease value of $k$ by $E(x,W)$. 
Note that in first branch we not only decrease value of $s'$ but we also increase a lower bound on 
$|N(x) \cap W|$ by $1$ as $|N(x) \cap W|\geq |W|-\frac{s}{2}$. 

Let us consider a path $(x_1, x_2, \dots x_l)$ from root to leaf in our branching tree. To each node 
we assign a vertex $x_i$ on which we are branching at this node. 
For each such path we construct unique sequence $a_1, a_2, \dots, a_m$ and a number  $b$. We put $b$ 
equal to the number of vertices from set $\{x_1, x_2, \dots, x_l \}$ that was assigned to solution $S$. 
And $a_i-1$ is a number of vertices that was assigned to $W$ in a sequence  $x_1, x_2, \dots x_j$ where $x_j$ is an $i-$th vertex assigned to $B$ in this sequence.
Note that $|N(x_j)\cap W|\geq a_i$, so $\sum_i a_i\leq k$. 
Note that for any path from root to leaf we can construct a corresponding sequence $a_i$ and number $b$. 
Moreover, any sequence $a_1, a_2, \dots a_m$ and number $b$ correspond to at most  one path  from root to node.  
  
\begin{proposition}\label{treeToSeq} 
Given number $b$ and non-decreasing sequence $a_1, a_2, \dots, a_m$ we can uniquely determine a corresponding path 
in a branching tree. 
\end{proposition}
\begin{proof}
For a notation convenience we let $a_0=1$. 
For $1\leq i \leq m$ we perform the following operation:
we make $a_i-a_{i-1}$  steps of assigning vertices to a solution set, 
i.e. to set $W$ and make one step in branch assigning vertex to a set $B$. 
After $m$ such iterations we perform $b-m$ steps of assigning vertices to solution.  
As $a_1, a_2, \dots a_m$ is non-decreasing sequence we have constructed a unique path in branching tree. 
It is easy to see that the original sequence $a_1, \dots , a_m$ and number $b$ 
correspond to a constructed path.  
So for each path from root to leaf there is a corresponding sequence and 
for each sequence with a number there is at most one corresponding path from root to node in a tree.
\end{proof}
  
\begin{lemma}\label{seqBound}
The number of tuples \((a_1, \dotsc, a_m, b)\) where 
\(0 \leq b \leq s\), \(1 \leq a_i \leq a_{i+1}\) for \(i < m\), and \(\sum_i a_i \leq k\) is bounded by \( \cOs\left(2^{\cO\left(\sqrt{k}\right)}\right) \)
\end{lemma}

\begin{proof}
For fixed \( l \), tuples \((a_1, \dotsc, a_m)\) such that \(\sum_i a_i = l\) are well-known and are called partitions of \(l\). 
Pribitkin~\cite{deAzevedoPribitkin2009} gave a simple upper bound
\(e^{2.57\sqrt{l}}\) on the number of partitions of \(l\). 
Hence, number of tuples \((a_1, \dotsc, a_m)\) is bounded by 
\(\sum\limits_{i = 0}^{k}e^{2.57\sqrt{i}} \leq (k + 1)e^{2.57\sqrt{k}}\).
Moreover, we know that \(0 \leq b \leq s\). It means that the number of tuples \((a_1, \dotsc, a_m, b)\) is bounded by \((s + 1)(k + 1)2^{\cO\left(\sqrt{k}\right)}\).
\end{proof}

From Proposition~\ref{treeToSeq} and Lemma~\ref{seqBound} follows that the number of nodes in 
a branching tree is at most $s2^{\cO\left(\sqrt{k}\right)}$. Hence, the running time of the procedure is 
at most $s 2^{\cO\left(\sqrt{k}\right)}$.
  
 Now, we compute required time for algorithm in this case(case 2). At first, we branch on a vertex and its neighbors from solution set $S$. We did it by  creating  at most $\cOs\left(2^{\cO\left(k^{2/3}\right)}\right)$ subcases. 
In each subcase we run a procedure with running time $\cOs\left(2^{\cO\left(\sqrt{k}\right)}\right)$. 
So, the overall runnning time equals to
$\cOs\left(2^{\cO\left(\sqrt{k}\right)}2^{\cO\left(k^{2/3}\right)}\right) = \cOs\left(2^{\cO\left(k^{\sfrac{2}{3}}\right)}\right)$.

The worst running time has \textbf{Case 1}, so the running time of the whole algorithms is \( \cOs\left(k^{\cO\left(k^{\sfrac{2}{3}}\right)}\right) \).
\end{proof}

\bibliography{hcc}

\end{document}